\documentclass[manuscript,nonacm]{acmart}
\usepackage{lipsum} 
\usepackage{combelow}

\setcopyright{none}
\copyrightyear{2021}
\acmYear{2021}
\acmDOI{}
\acmConference[]{}{}{}
\acmBooktitle{}
\acmPrice{}
\acmISBN{}

\begin{document}

\title{Public Communication can Facilitate Low-Risk Coordination under Surveillance}

\author{Amos Korman}\thanks{Contact author: amos.korman@irif.fr}
\email{amos.korman@irif.fr}
\affiliation{%
  \institution{CNRS, IRIF, Universit\'e de Paris}
  \city{75013 Paris}
  \country{France}
}

\author{Pierluigi Crescenzi}
\email{pierluigi.crescenzi@gssi.it}
\affiliation{%
  \institution{Gran Sasso Science Institute}
  \city{67100 L'Aquila}
  \country{Italy}
}

\begin{abstract}
\textbf{Abstract.} Consider a sub-population of rebels that wish to initiate a revolution. In order to avoid initializing a failed revolution, rebels would first strive to estimate their relative "power", which is often correlated with their fraction in the population. However, and especially in non-democratic countries, rebels refrain from disclosing themselves. This poses a significant challenge for rebels: estimating their fraction in the population while minimizing the risk of being identified as rebels. This paper introduces a distributed computing framework aiming to study this question. Our main takeaway message is that the communication pattern has a crucial role in achieving such a task. Specifically, we show that relying on the inherent noise in the communication, \emph{public communication}, characterized by the fact that each message announced by an individual can be viewed by all its neighbors, allows rebels to estimate their fraction in the population while keeping a negligible risk of each rebel being identified as such. The suggested estimation protocol, inspired by historical events, is extremely simple and can be executed covertly even under extreme conditions of surveillance. Conversely, we show that under peer-to-peer communication, protocols of similar simplicity are either inefficient or non-covert.
\end{abstract}

\maketitle

\section{Introduction}

Large scale changeovers in a population, such as political overthrows of dictatorships by rebels, are often perceived as complex emergent phenomena \cite{gurr2015men,howard2013democracy,siani2005romanian,dinges2005condor}. In order to avoid conducting a failed rebellion, rebels would typically refrain from initiating a revolution until they manage to obtain a reliable indication that many others will join them \cite{centola2018experimental,Popovic2015}. In other words, rebels would first try to understand whether ``\textit{we} are the \textit{many} and \textit{they} are the \textit{few}''~\cite{Shelley1832}. However, and especially under non-democratic regimes, obtaining such information may not be a trivial task, as rebels often refrain from disclosing themselves as such. In turn, with the lack of such information, the status quo may continue to hold even when the support for a revolution is effectively high. Nevertheless, in some cases, rebels successfully coordinate their actions even under severe conditions of surveillance. 

Drawing on distributed computing reasoning, we argue that the pattern of communication can play a key role in the ability of rebels to reach such a coordination safely. We distinguish between two extreme patterns in a network environment:
\begin{itemize}
\item {\em Private communication}: Each message sent by an agent is heard by a single neighbor.

\item {\em Public communication}: Each message sent by an agent is heard by all its neighbors. 
\end{itemize}
Private communication aims to capture one-to-one interactions, which are executed either directly in person or via a digital private messaging platform. In contrast, public communication aims to model social media infrastructures, such as Facebook or Twitter, that played a significant role during the Arab spring revolutions \cite{howard2013democracy}, behavioral communication, such as the ``slow-motion'' day organized by the opponents of the Pinochet dictatorship in the 1970s~\cite{Popovic2015}, crowd assemblies, such as the one gathered during the last public speech by Ceau\cb{s}escu in 1989 \cite{siani2005romanian}, or even chemical communication during quorum sensing by pathogenic bacteria before attacking their host \cite{quorum2}. The later two examples are discussed in more details in Section \ref{sec:discussio}.

The parallel nature of public communication can allow for fast information spread. This is commonly considered a fundamental property that explains how the use of social media helped accelerate several social movements \cite{howard2013democracy}. Here we argue that from the perspective of rebels, there is another significant benefit of public communication: It not only facilitates fast information spread but also allows for its covert dissemination. More specifically, we argue that public communication facilitates the ability of agents belonging to a sub-population to estimate their proportion in the population without revealing that they belong to the sub-population.

To obtain some intuition regarding the difficulties involved in the rebels estimating their proportion in the population without revealing that they are rebels, let us consider a simple setting with two communicating agents: Alice and Bob. A possible scenario can be the following. Being a rebel, Alice would try to understand whether Bob is also a rebel. While talking normally, she could start by ``cautiously tempting'' Bob into sending ``rebellious'' signals. If Bob would be a rebel, then, in turn, he may ``cautiously respond'' to Alice by sending some, but not too many, rebellious signals, and, in parallel, try to ``cautiously tempt'' Alice into doing the same. At the end of the conversation, each person would classify the other as a rebel if the (weighted) number of rebellious signals he or she received passes a certain threshold. Unfortunately, in such a scenario, unless employing some sophisticated cryptographic protocol (which is highly unlikely in direct communication between humans), there is little hope for rebels if the police are surveilling all conversations. Indeed, if both rebels could detect that the other person is a rebel, e.g., by counting the weighted rebellious signals he or she sends, then so could the police. In fact, the same argument holds with respect to any evaluation criterion used by one of the parties that takes as input only the conversation between the parties (and not, for example, some random private key generated by a party before the execution starts as could be done using cryptographic techniques \cite{yao1982protocols}). However, these difficulties do not rule out the possibility that simple covert estimation mechanisms exist in a multi-party scenario. Indeed, in contrast to the two-party scenario, understanding that there are many rebels in a large population does not necessarily imply that one can identify who they are.

\subsection{Model}

This paper introduces a distributed computing framework that aims to study covert population-size estimation by humans or other biological entities. For this purpose, we give special attention to simplicity, in both message encoding, and decoding. In particular, we assume that messages are real numbers that, in the context of revolution, represent a certain level of satisfaction from the ruling entity. In turn, the decoding is assumed to be threshold-based, capturing the ability to sense a certain tendency. These assumptions are in contrast to cryptographic schemes, that are based, for example, on first generating a huge random number and then manipulating it in a sophisticated manner \cite{yao1982protocols}. 

Formally, we consider an idealized model consisting of $n$ agents, communicating over a network $G$, where the nodes  represent the agents and the edges represent communication links between neighbors. The {\em degree} of an agent $i$, denoted $\Delta_i$, is the number of neighbors of $i$ in $G$. Let $\Delta$ denote the median degree.

Initially, each agent is chosen as a {\em rebel} with probability $0\leq \rho\leq 1$, independently of others. Otherwise, it is an {\em obedient citizen}.  A priori, the  behavior type of an agent is known to itself but not to others. The parameter $\rho$, which can be considered as the fraction of rebels in the population, is unknown to the agents. We say that there are ``many rebels'' if $\rho \geq 0.8$, and ``few rebels'' if $\rho \leq 0.2$, where it should be clear that the constants $0.8$ are $0.2$ are  arbitrarily and any other constants $1>c_1>c_2>0$ could have been used instead. Informally, the aim of the rebels is to distinguish the case of many rebels from the case of a few rebels while minimizing the risk of disclosing themselves as rebels.

As mentioned, we distinguish between public and private communication models. In either case, we assume that communication consists of a single round, in which agents exchange messages in parallel, so that each agent has access to (a distorted version of) a message sent by each of its neighbors. The restriction to one round is made for the sake of simplicity of definitions, however, it should be clear that our framework can be extended to multiple rounds. 

Each message is modeled as a real number, which may encode information about the level of satisfaction with the ruling entity. For simplicity and normalization, we assume that the messages sent by obedient citizens are always equal to $0$. In contrast, a rebel may freely choose the messages it sends. Unless mentioned otherwise, we consider only deterministic protocols. Moreover, when considering the private communication model, we restrict attention to {\em uniform} protocols, in which the same message is sent to all neighbors. Such protocols are natural analogs of protocols in the public communication model. Hence, in both private and public models, a rebel deterministically decides on  a single message $m$ to be delivered to  its neighbors. The difference between the models, is that in the private model an agent $i$ actively sends $\Delta_i$ copies of $m$ (one copy per neighbor), whereas in the public model it only announces the message $m$ once, and then this message is heard by all its neighbors. 

Importantly, in both models, the receiver of a message may not interpret the corresponding information correctly. To capture this, we assume that every message $m_{i,j}$, originated at agent $i$ and heard by $j$, is received by $j$ as 
\[s_{i,j} = m_{i,j} +N(0,1),\]
where $N(0,1)$ is a normally distributed variable. This noise variable is sampled for each neighbor $j$ of agent $i$, independently from all other neighbors of $i$. In order to avoid confusion between a message $m_{i,j}$ and its distorted version $s_{i,j}$, we refer to the latter as a {\em signal}. 

At the end of the communication round, after receiving the signals from all its neighbors, each rebel either outputs ``many'', or does not output anything. The {\em success probability} of a rebel protocol, denoted $p_{\mathrm{success}}$, is the probability that at least a third of the rebels output ``many'' when $\rho\geq 0.8$ (as before, the constant $1/3$ is arbitrary). The {\em output-risk} of a rebel, denoted $r_{\mathrm{output}}$, is defined as the probability that it outputs ``many'', given that $\rho\leq 0.2$. This captures the risk of having a false positive. 

Another component of the risk corresponds to being identified as a rebel as a result of sending too suspicious messages. This risk depends not only on the messages sent by a rebel but also on the abilities of the surveilling entity, called {\em police}, and on the criteria it uses to identify rebels. Aiming to capture extreme conditions associated with totalitarian countries, we assume that the police surveils all communication links. However, for fairness considerations, similarly to the agents in the system, the police cannot see the actual messages sent and instead sees their corresponding signals. That is, every message $m$ sent by an agent is seen by the police as a signal $s = m +N(0,1)$, where the sample noise is independent of all other events.  For each agent $i$, the police protocol considers the signals associated with all the messages  $i$ sends and then decides whether or not to {\em arrest} the agent. Note that under public communication, the police receives one signal from each agent $i$ (since it announces one message only), whereas in the private communication model it receives $\Delta_i$ such signals. 

The police's goal is to arrest as many rebels as possible while minimizing arrests of obedient citizens. Being permissive with respect to the police, we assume that it's computational power is unlimited, and that it knows both the rebel protocol and the fraction of rebels $\rho$. Conversely, being restrictive with respect to rebels, we assume that rebels do not know the police's protocol, and must guarantee low risk with respect to any police protocol. Specifically, the \textit{relative message-risk} of a rebel protocol, denoted $r_{\mathrm{message}}$, is defined as the maximal difference between the probability that a rebel is arrested by the police and the probability that an obedient citizen is arrested, where the maximum is taken over all police protocols. Finally, the {\em total risk} of a rebel is the output-risk plus the relative message-risk:
\[r_{\mathrm{total}}=r_{\mathrm{output}}+r_{\mathrm{message}}.\]

To illustrate the definitions, let us briefly discuss two trivial protocols operating under the public communication model. In the first protocol each rebel outputs ``many'' regardless of the messages it receives. There, the success probability is extremely high, namely 1, but so is the output-risk. A second trivial protocol imitates the behavior of obedient citizens by letting each rebel announce the message zero. This protocol has relative message-risk of zero, but regardless of its outputting rule (concerning when to output ``many''), it cannot maintain both high success probability and low output-risk. 

In summary, the goals of the rebels are (1) to maximize success probability, and (2) to minimize the total risk.

\subsection{Our results}
This paper introduces a distributed computing framework that aims to study covert computations by humans or other biological entities. Our focus is on the ability of a sub-population of agents (rebels) to perform very simple computations to estimate their fraction in the population, while minimizing the risk of  exposing the fact that they belong to the sub-population.  The model assumes that other agents (obedient citizens) simply send the number zero to their neighbors, which in turn receive a distorted version of this number due to noise. Our main takeaway message is that even under extreme surveillance conditions, there are simple deterministic protocols in the public communication model that allow rebels to estimate their fraction in the population while keeping a negligible risk of each rebel being identified as such. Conversely, we show that under a peer-to-peer analogue, protocols of similar simplicity are either inefficient or non-covert. We next describe our results in more details.

\subsection*{The {\em Quorum-Sensing} protocol}

We first consider an extremely simple rebel protocol, termed {\em Quorum-Sensing}, which is particularly useful when executed in the public communication model (see Figure~\ref{fig:quorumsensing}). The protocol is inspired by historical events that happened during the Pinochet dictatorship in the 1970s. The rebels opposing Pinochet used the idea of suggesting to act slowly, for example, that taxi drivers would drive slower than usual. The message spread rapidly, and many people cooperated in this initiative~\cite{Popovic2015}. Watching the city's low motion, the rebels could realize that they were many without incurring considerable risk.

Formally, in this protocol, each rebel simply sends the message~$m=\epsilon$, for some predetermined parameter $\epsilon>0$. At the end of the communication round, a rebel outputs ``many'' if and only if (1) its degree is at least the median degree $\Delta$, and (2) the average value of a signal received from a neighbor is at least $\epsilon/2$. \begin{figure*}[h]
\centering
\includegraphics[scale=0.15]{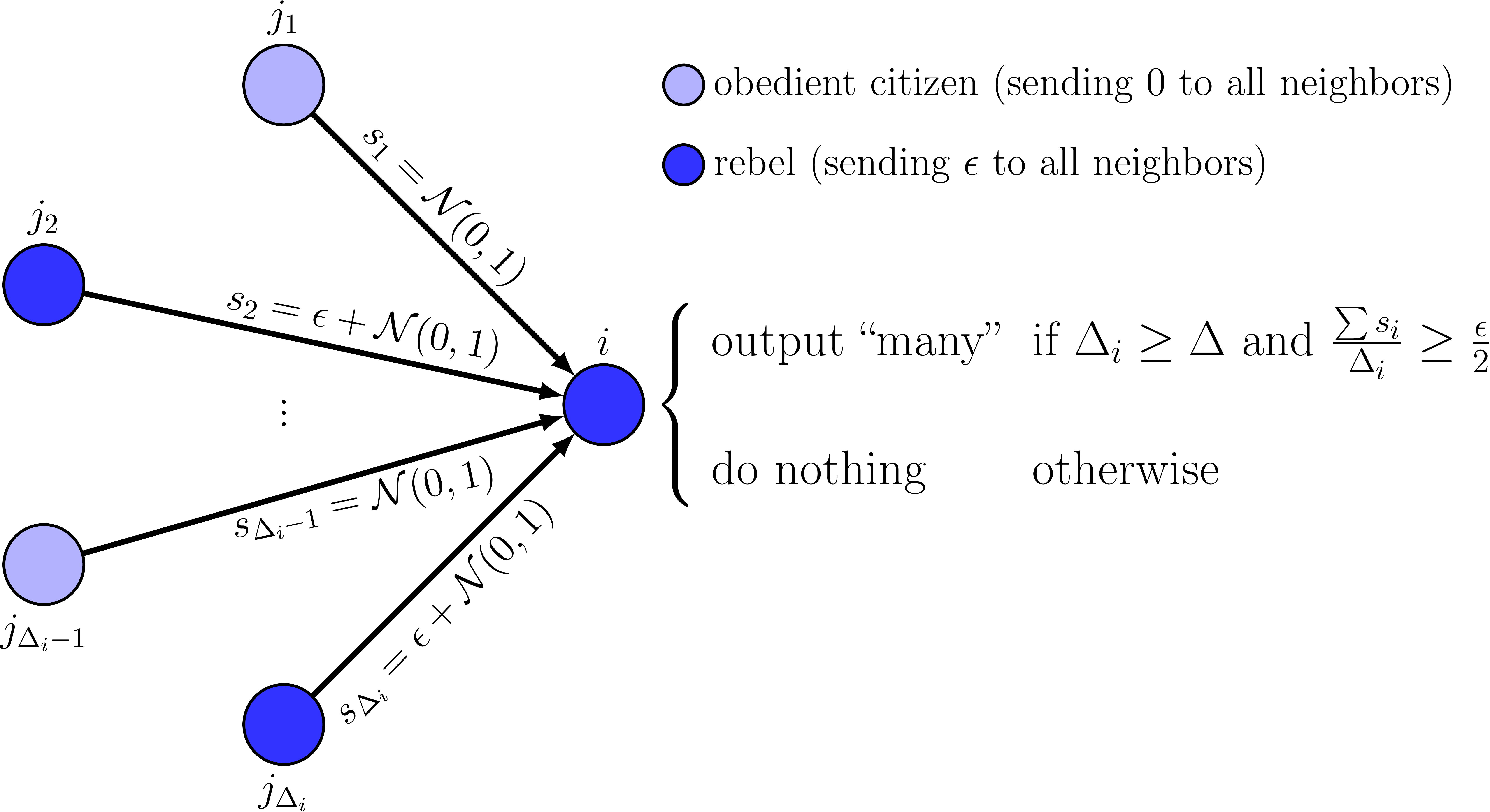}
\caption{\textbf{The Quorum-Sensing protocol with parameter $\epsilon$.} 
Rebels send the number $\epsilon>0$ to all their neighbors, and obedient citizens send the number 0. Rebel $i$ receives signals from its $\Delta_i$ neighbors (a combination of rebels and obedient citizens), 
and decides whether to output ``many'' according to the rule specified on the right ($\Delta$ denotes the median degree).}
\label{fig:quorumsensing}
\end{figure*}

The next theorem states that under public communication, $\epsilon$ can be set to be sufficiently small to guarantee that the total risk incurred by this protocol is arbitrarily small, while still maintaining extremely high success probability on highly connected networks.  

\begin{theorem}\label{thm:success}
Consider the public communication model and a network with median degree $\Delta$. For any $\epsilon>0$ and for $n$ and $\Delta$ sufficiently large, the success probability of the Quorum-Sensing protocol is at least $1-O(1/n^2)$, while the total risk is at most $0.715\epsilon$.
\end{theorem}
The formal proof of the theorem is given in Section \ref{sec:success}. Intuitively, the reasoning behind the proof is as follows. In public communication, each agent sends only one message. If this message is sufficiently close to what an obedient citizen sends, i.e., close to 0, then the rebel can hide behind the noise (see Lemma \ref{thm:safety}). On the other hand, when observing the signals coming from many agents, a small bias in the original messages of many rebels becomes visible, due to the law of large numbers that effectively cancels noise. 

We corroborated this result by conducting simulations of the Quorum-Sensing protocol over a real-world social network. We considered the Facebook graph released in~\cite{GjokaKBM10,GjokaKBM11}, which was collected in April 2009, containing a sample of approximately 1.2 million users reached by one breadth-first-search traversal. The results are shown in Figure~2.

\begin{figure*}[h]
\centering
\includegraphics[width=.9\linewidth]{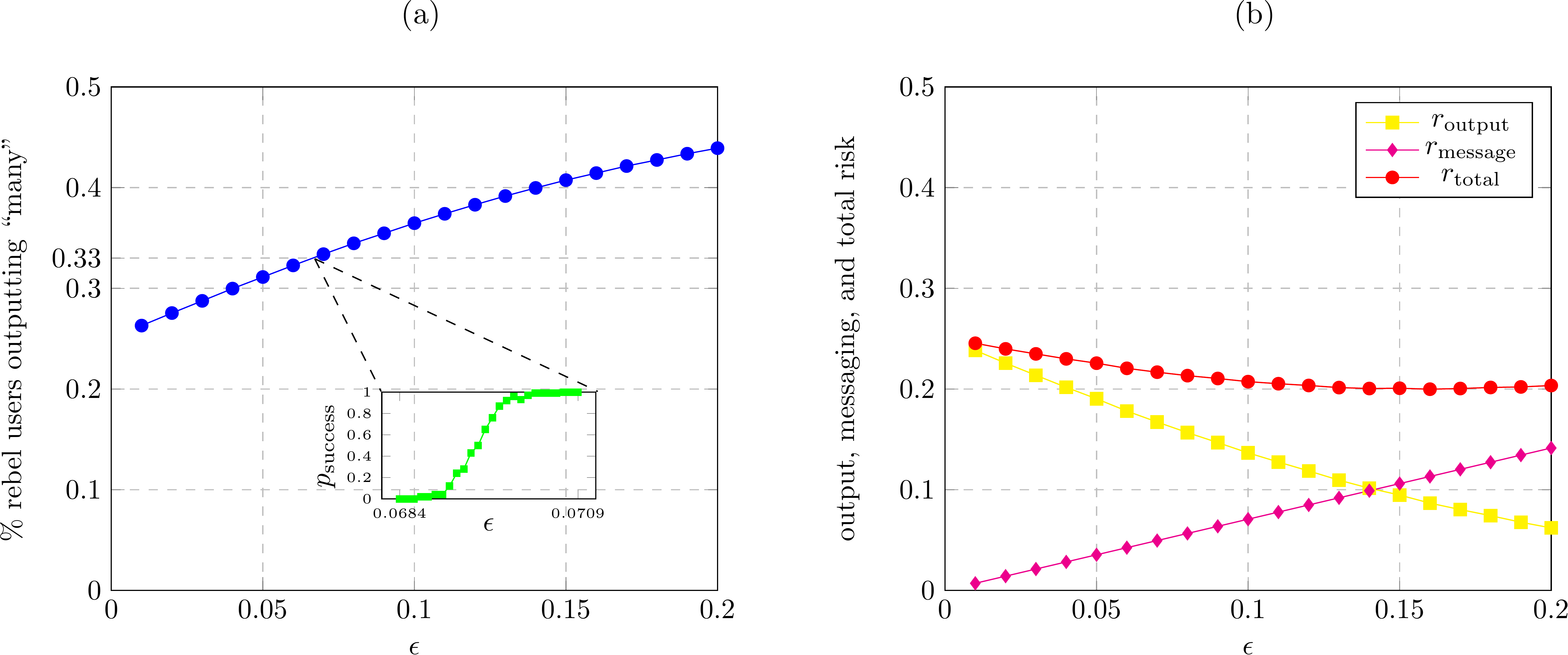}
\caption{\textbf{Public communication: Simulations of the Quorum-Sensing protocol on a Facebook network.} The percentage of rebels which output ``many'' (a) and the risk of a rebel (b) as a function of the protocol parameter $\epsilon$ in the case of a Facebook sub-network. The inset in (a) represents the success probability in the interval between $0.0684$ and $0.0709$. From (b) we can conclude that the minimum total risk $r_\mathrm{total}$ is approximately $0.2$ and is obtained at roughly $\epsilon=0.16$, which corresponds to a success probability of almost $1$.}
\label{fig:facebook-public}
\end{figure*}

\subsection*{Robustness to undercover agents}

In some scenarios, {\em undercover agents} may secretly cooperate with the police. In the context of our model, such agents aim to distort the detection protocol of the rebels to reduce the rebels' success probability or to increase the rebel's output-risk. For example, the Quorum-Sensing protocol's correctness is very sensitive to undercover agents. In the case of a complete network, e.g., even a single undercover agent can diminish the correctness of the Quorum-Sensing protocol by sending a message consisting of a huge number. 

We next present a variant of the Quorum-Sensing protocol,  called {\em Median}, which is robust to a non-negligible fraction of undercover agents. The protocol uses the same messaging protocol as the Quorum-Sensing protocol, i.e., it deterministically sends the message $\epsilon$. Moreover, similarly to the Quorum-Sensing protocol, a rebel $i$ outputs nothing if its degree is small, i.e., if $\Delta_i<\Delta$. Otherwise, it counts the number of incoming signals that are above $\epsilon$, and outputs ``many'' if and only if the number of such signals exceeds $(\frac{1}{2}-\frac{7\epsilon}{30})\Delta_i$. 

The next theorem quantifies the robustness of the Median protocol to the presence of undercover agents. It suggests that for a range of relatively small $\epsilon$, the Median protocol yields similar guarantees as the Quorum-Sensing protocol, even when facing a small fraction of undercover agents. The proof of the theorem is given in Section \ref{sec:median}.

\begin{theorem}\label{thm:median}
Assume that the probability that an agent is undercover is $o(1)$. If $\epsilon\in[0.04,0.2]$, then for $\Delta$ sufficiently large, the success probability of the Median protocol is at least $1-O(1/n^2)$, and the total risk is at most $0.715\epsilon$.
\end{theorem}

The ``o'' notation in the theorem is with respect to the median degree $\Delta$. We further note that the particular constants we are going to use here are not meant to be optimized. Instead, these constants are used for convenience, as they are based on specific bounds on the tail distribution of a normal distribution. 

\subsection*{An Impossibility Result under Private Communication}

In the context of private communication, each rebel executing the Quorum-Sensing protocol (or the Median protocol) would send the number $\epsilon>0$ to each neighbor. Hence, instead of sending just one message $\epsilon$ as in the public communication case, a rebel $i$ now sends $\Delta_i$ such messages. The next theorem states that any such uniform deterministic protocol fails to provide both low total risk and high success probability, regardless of the decision of when to output ``many''. The proof of the theorem appears in Section \ref{sec:impossibility}.

\begin{theorem}\label{lower-independent}\label{thm:impossibility}
Consider a private communication framework in a regular network of degree $\Delta$ of size $n$. Consider any uniform deterministic rebel protocol. Assume that the success probability is bounded away from zero for sufficiently large $n$ and $\Delta$, that is, the success probability is at least $p$, for some constant $p>0$. Then the total risk of a rebel is at least $\frac{p}{4}-\frac{1}{n}$.    
\end{theorem}

To compare with the simulations on the Quorum-Sensing protocol on the Facebook sub-network in Figure~\ref{fig:facebook-public}, we simulated this protocol on the same network, but under private communication instead of public communication. Confronting the rebels, we used the {\em Reverse police protocol}, which intuitively uses the rebel's decision protocol to decide when to arrest an agent. The results of the simulations are presented in Figure~\ref{fig:facebook-private}. As expected, the output-risk is the same as in the public communication model since the signals outgoing from an agent follow the same distribution in both models, and hence, the rebels' output follows the same distribution. For the same reason, the percentage of rebels which output ``many'' in case $\rho\geq 0.8$, and the success probability are the same as under the public communication model, as presented in Figure~\ref{fig:facebook-public}(a). However, the relative message-risk and, hence, the total risk of a rebel are significantly higher under private communication than the ones observed under public communication (see Figure~\ref{fig:facebook-public}(b)), despite the fact that the latter are obtained against any police protocol. Indeed, since the success probability is very close to zero when $\epsilon\leq 0.0684$ (inset in Figure~\ref{fig:facebook-public}(a)), the interesting cases are when $\epsilon>0.0684$. In this range, the total risk under private communication is at least $0.4$, which is about twice the total risk under the public communication model for the same range of $\epsilon$.

\begin{figure}[h]
\centering
\includegraphics[scale=0.25]{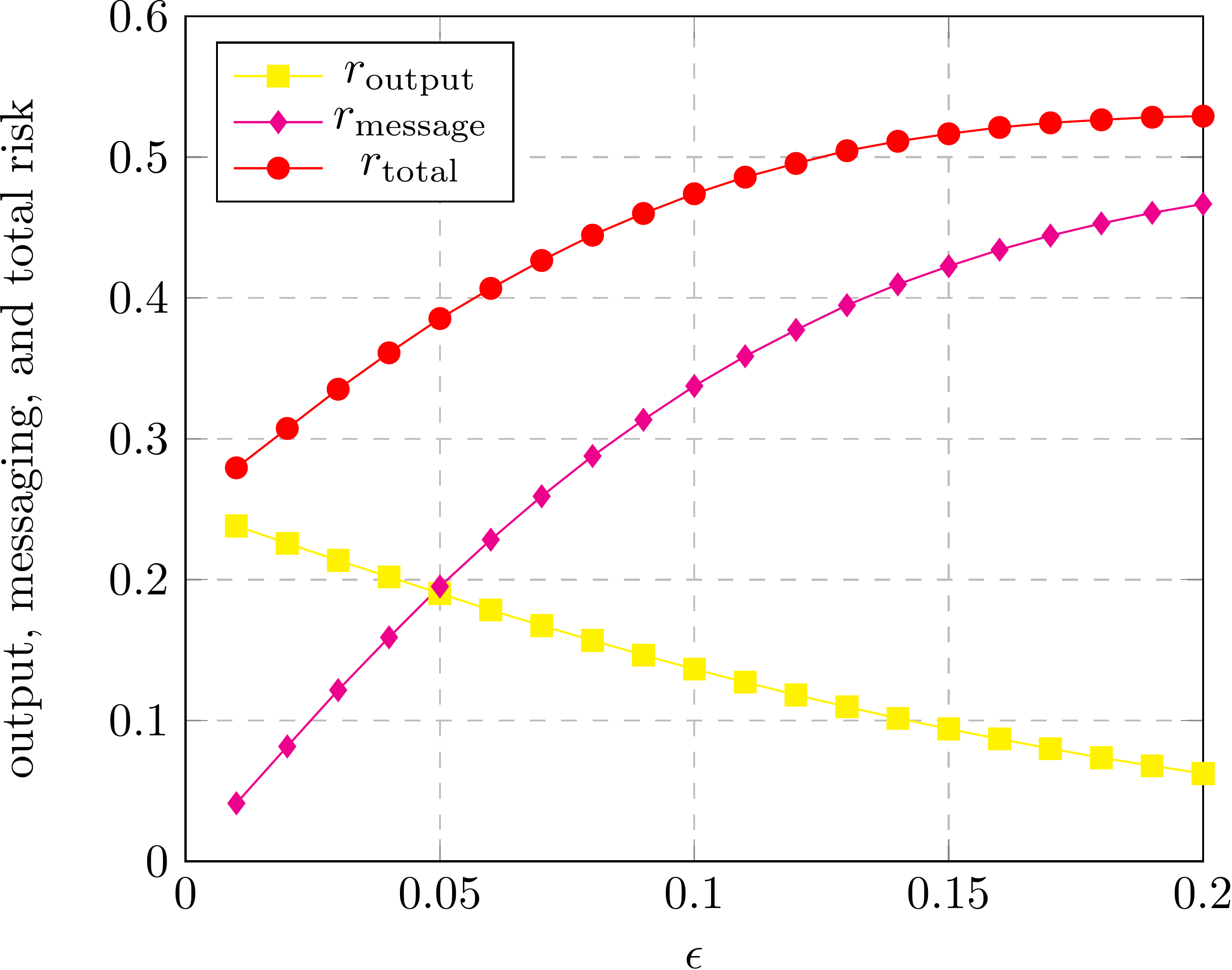}
\caption{\textbf{Private communication: The risk of the Quorum-Sensing protocol.} The plots depict the message, output, and total risks incurred by a rebel running the Quorum-Sensing protocol in the private communication model, when facing the Reverse police protocol. The simulations should be compared with the simulations shown in Figure~\ref{fig:facebook-public}(b).}
\label{fig:facebook-private}
\end{figure}

\subsection*{The {\em Self-Immolation} protocol}

As mentioned earlier, there is no deterministic uniform rebel protocol in the private communication model that is both efficient and secure. Nevertheless, we show that in this model, there exists a randomized uniform protocol, termed {\em Self-Immolation}, that achieves both high success probability and low risk. Importantly, however, the low risk is guaranteed only {\em on average}, and, in fact, the success probability of this protocol depends on few rebels that ``consciously sacrifice'' themselves, in the sense that they send messages that would clearly distinguish them from obedient citizens, and consequently put them at a very high risk. The proposed protocol is inspired by several historical events corresponding to the emergence of revolutions, including the self-immolation of Mohamed Bouazizi on 17 December 2010 relating to the Tunisian revolution, and the self-immolation of Th\'ich Quang Duc on 11 June 1963, related to the Buddhist crisis in Vietnam. 

The Self-Immolation protocol is as follows. Each rebel tosses a coin, and with probability $q$ sends a huge number $M=\infty$ to all of its neighbors; otherwise, with probability $1-q$, it sends the message~0. In turn, at the end of the communication round, a rebel $i$ outputs ``many'' if and only if (1) its degree is $\Delta_i\geq \Delta$ and (2) it sees more than $\tau\cdot \frac{\Delta_i}{\Delta}$ messages containing $M$, for some threshold~$\tau$. 

The next theorem states that in the private communication model, the Self-Immolation protocol achieves both high success probability and low risk. The formal proof is given in Section \ref{sec:analyze-self}. The intuition is that if a rebel sees a very large number then, since this event is so rare, there must be many rebels.

\begin{theorem}\label{thm:self-mmulation}Consider  the Self-Immolation protocol with $q=c\log n/\Delta$, and $\tau= (c\log n)/2$ for a sufficiently large constant $c$. Consider a network such that $\Delta\gg \log n$. Then, the average risk is $O(\log n/\Delta)$ and the success probability is $1-O(1/n^2)$.
\end{theorem}

\subsection{Related works}

In general, the ability of rebels to perform covert computations  depends, among other things, on both their own computational abilities and the ones of the surveilling entity. Intuitively, the more powerful the surveilling entity is the more similar to ordinary civilians rebels should appear. When the surveilling entity is unlimited in its computational power, the tools to analyze the ability of rebels to covertly do computations naturally point to information theory. Of particular relevance is the recently introduced notion of {\em information complexity} \cite{braverman2013information,ma2011some,ma2013infinite,filmus2019information}, which quantifies the amount of information that the communication reveals regarding the inputs of the players to each other, or to an external observer. In \cite{braverman2013information}, Braverman et al.~showed that the information complexity of the AND function in a two-party multi-round setting, assuming that errors are not allowed, is $1.4923$ bits. This result appear to suggest that if the surveilling entity is unlimited in its computational power, then highly covert computation with small errors in a two-party scenario cannot be achieved even in a multi-round setting.

Secure computation under cryptographic assumptions was introduced by Yao \cite{yao1982protocols}. By now, there is a huge body of literature on secure computations, in both two-party scenarios and multi-party scenarios, including the case of tolerating a quorum of Byzantine players \cite{chen2007secure}, and making known distributed algorithms secure \cite{parter2019secure}. The concept of covert computation was introduced in \cite{von2005covert} for two-party scenarios and in \cite{chandran2007covert} for multi-party scenarios. The idea behind covert protocols is that parties do not know if other parties are also participating in the protocol or not. In general, however, most of schemes in the cryptography literature employ sophisticated operations on both the encoding and the decoding parts. While such operations can be implemented by computers, they cannot be expected to be employed directly by humans or other biological entities. 

Counting events occurring at different places in a network, and, in particular, estimating the population size are fundamental problems in distributed computing \cite{afek1996local,emek2011new,kuhn2010distributed,DBLP:conf/podc/EmekK10,beauquier2015space,izumi2014space}. Here, we are mostly interested in simple estimation protocols, based on sampling and sensing a certain tendency. Such mechanisms are natural for biological entities. For example, similar protocols are executed by bacteria communities aiming to identify when their density passes a certain threshold \cite{quorum1} (see more details in Section~\ref{sec:discussio}).

\section{Analysis of the Quorum-Sensing protocol}\label{sec:success}

The goal of this section is to prove Theorem \ref{thm:success}. Towards this, we first obtain the following lemma that establishes the low relative message-risk of the Quorum-Sensing protocol. Intuitively, its proof is based on first showing that the ``distance'' between the distributions of the signals whose mean is $\epsilon$ (corresponding to the messages out-going from a rebel) versus those whose mean is 0 (corresponding to the messages out-going from an obedient citizen) is small, and then deducing that any police protocol that attempts to distinguish them must make many mistakes. In order to prove this we rely on concepts and techniques adopted from the area of {\em statistical hypothesis testing}. 

\begin{lemma}\label{thm:safety}
Consider the public communication model, and assume that each rebel announces the message~$\epsilon>0$. Any police protocol that arrests a rebel with probability at least~$\delta$, must arrests an obedient citizen with probability at least~$\delta-\epsilon/\sqrt{2}$. In other words, the relative message-risk of a rebel is at most $\epsilon/{\sqrt{2}}$. This holds with respect to any underlying graph.
\end{lemma}

\begin{proof}
Recall that each obedient citizen sends the number 0 which is then received as a signal $a\sim N(0,1)$, distributed normally with mean 0 and variance 1. A rebel sends  the message $\epsilon$ which will in turn be received as a signal $b\sim N(\epsilon,1)$. Our goal is to prove that any police protocol that arrests a rebel with probability at least~$\delta$, must arrest an obedient citizen with probability at least~$\delta-\epsilon/\sqrt{2}$. To this aim, we will, intuitively, show that the ``distance'' between the signal distribution of an obedient citizen and that of a rebel is small, and then deduce that any police protocol that attempts to distinguish them must make many mistakes.

To formalise the notion of distance between two distributions, we use two measures which are standard in the field of statistical hypothesis testing. Given two distributions $P_0$ and $P_1$ supported on the real numbers, we define the following.
\begin{itemize}
\item {\em Total variation distance}
\[\mathrm{TV}(P_0,P_1) :=  \frac 12 \int_{-\infty}^{\infty}|p_0(x) - p_1(x)|dx
\]

\item{\em Kullback-Leibler divergence}
\[
\mathrm{KL}(P_0,P_1) := \int_{-\infty}^{\infty} p_0(x) \log \frac{p_0(x)}{p_1(x)}dx
\]
\end{itemize}
Here, $p_i(x)$ stands as a shorthand for $P_i(X=x)$, for $i=0,1$, and the logarithm is the natural one.

A configuration of messages $C$ is a vector of messages, $C=(m_1,m_2,\ldots,m_n)$, where $m_i$ is the message sent by agent $i$. Consider any police protocol $g$. Being liberal, we assume that the police can decide whether to arrest an individual $i$, based on all signals it receives from all agents, and not only on the one corresponding to the outgoing message from $i$. Hence, $g$ can be viewed as a (possibly probabilistic) function $g:R^n\rightarrow\{0,1\}^n$, deciding for each possible set of signals whether to arrest agent $i$, with $1\leq i\leq n$. Fix an agent $i$ and fix the messages $C_{\neq i}=(m_1,m_2,\ldots,m_{i-1},m_{i+1},\ldots, m_n)$ sent by all other agents. Let $g_{i,C_{\neq i}}:R\rightarrow\{0,1\}$ be the restriction of $g$ to agent $i$, given that all others sent $C_{\neq i}$. For brevity, we omit the subscripts $i$ and $C_{\neq i}$. Let $X$ be a random variable representing the value of the signal received by the police corresponding to the message sent by agent $i$, i.e., $X$ is drawn from either $P_0$ or $P_1$. If the police believes that $X$ corresponds to a message sent by an obedient citizen, then $g(X)=0$ (and the sender $i$ is not arrested), and if it believes that $X$ corresponds to a message sent by a rebel, then $g(X)=1$ (and $i$ is arrested). Our goal is to show that $g$ necessarily makes many mistakes. 

The next lemma shows that if the total variation distance between $P_0$ and $P_1$ is small, then any police function $g$ can distinguish obedient citizens from rebels only with a very small probability.

\begin{lemma}[{{Neyman-Pearson~\cite[Lemma $4.3$ and Proposition $4.4$]{rigollet2015high}}}]
\label{lem:neyman}
Let $P_0$ and $P_1$ be two distributions with support $R$. Let $X$ be a random variable drawn from either $P_0$ or $P_1$. Consider a (possibly probabilistic) mapping $g:R\rightarrow\{0,1\}$ that attempts to ``guess'' whether the observation $X$ was drawn from $P_0$ (in which case it outputs 0) or from $P_1$ (in which case it outputs 1). Then, we have that
\[
P_{0}\left(g(X)=1\right) + P_{1}\left(g(X)=0\right)\geq 1-\mathrm{TV}(P_0,P_1),
\]
where $P_{i}\left(g(X)=j\right)$ stands for $P(g(X)=j\mid X\in P_i)$, for $i,j\in\{0,1\}$.
\end{lemma}

Lemma \ref{lem:neyman} bounds the probability of both kinds of ``wrong'' police actions: the arrest of an obedient citizen, and the non-arrest of a rebel. Our goal now is to bound from above the total variation distance of the two distributions considered. To this end, we use the next connection between the total variation distance and the Kullback-Leibler divergence.

\begin{lemma}[{{Pinsker~\cite[Lemma $4.8$]{rigollet2015high}}}]
\label{lem:pinsker}
For any two distributions $P_0$ and $P_1$,
\[
\mathrm{TV}(P_0,P_1) \leq \sqrt{\mathrm{KL}\left(P_0,P_1\right)}.
\]
\end{lemma}
As, in our case, $P_0$ and $P_1$ are normal distributions, computing their Kullback-Leibler divergence is easy:
\[
\mathrm{KL}(P_0,P_1) = \int_{-\infty}^{\infty} \frac{1}{\sqrt{2\pi}}\exp(-x^2/2)\log \frac{\exp(-(x-\epsilon)^2/2)}{\exp(-x^2/2)} dx = \epsilon^2/2. 
\]

By Lemma~\ref{lem:pinsker}, it follows that $\mathrm{TV}(P_0,P_1)\leq \epsilon/\sqrt 2$.
	
Fix an execution of the protocol, and the messages of all players. For a signal received by the police from an agent, Lemma~\ref{lem:neyman} implies that the probability of mistake by the police is at least $1-\epsilon/\sqrt 2$, i.e., $P_{0}\left(g(X)=1\right)+P_{1}\left(g(X)=0\right)\geq 1-\epsilon/\sqrt 2$. If the police arrests a rebel with probability at least~$\delta$, i.e., $P_{1}\left(g(X)=1\right) \geq\delta$, then $P_{1}\left(g(X)=0\right)\leq1-\delta$. This implies that $P_{0}\left(g(X)=1\right)\geq\delta-\epsilon/\sqrt 2$, i.e., the police arrests an obedient citizen with probability at least~$\delta-\epsilon/\sqrt 2$. This completes the proof of Lemma \ref{thm:safety}.
\end{proof}

We are now continue with the proof of Theorem \ref{thm:success}. As stated in Lemma \ref{thm:safety}, for a very small $\epsilon>0$, the rebel messaging protocol incurs very small relative message-risk. To complete the proof of Theorem \ref{thm:success} we show that the total-risk is also very small, and that, nevertheless, the effectiveness of the Quorum-Sensing protocol is extremely high. The proof is based on simple applications of Chernoff's inequality.

Consider a rebel $u$ with degree $\Delta_u$. We first aim to estimate the output-risk of $u$, that is, $r_\mathrm{output}$. If $\Delta_u<\Delta$, then $u$ does not output anything and, hence, $r_\mathrm{output}=0$. Let us therefore assume that  $\Delta_u\geq \Delta$. Denote by $\bar{s}=\sum s_i/\Delta_u$  the average value of the signals $u$ receives. Note that $\bar s$ is a random variable drawn from $N(\mu,\sigma^2)$, where $\mu=\rho\epsilon$ and $\sigma^2=1/\Delta_u$. By applying the Chernoff bound for a normal distribution, we have that, for any $\rho\leq0.2$, the output-risk is
\[
r_\mathrm{output} = P\left(\bar{s}\geq\epsilon/2\right)\leq P({\bar s-\mu}\geq 3\epsilon/10) \leq P\left({\mid \bar s-\mu\mid}\geq\sigma\frac{3\epsilon\sqrt{\Delta_u}}{10}\right)\leq 2e^{-\frac{9\Delta_u\epsilon^2}{200}}.
\]
By Lemma \ref{thm:safety}, the relative message-risk $r_{message}$ is at most $\epsilon/\sqrt{2}$. Hence, the total risk is
\[
r_\mathrm{total}=r_\mathrm{output}+r_\mathrm{message}\leq\frac{\epsilon}{\sqrt{2}}+2e^{-\frac{9\Delta_u\epsilon^2}{200}}\leq\frac{\epsilon}{\sqrt{2}}+2e^{-\frac{9\Delta\epsilon^2}{200}}.
\]
Next, let us calculate the success probability. For this purpose, we need to lower bound the probability that at least a third of the rebels output ``many'' when $\rho\geq 0.8$. Consider a rebel $u$ with degree at least $\Delta$. By applying the Chernoff bound for a normal distribution, we have that, for any $\rho\geq0.8$, the probability that $u$ does not output ``many'' is
\[
P\left(\bar{s}<\epsilon/2\right) \leq P(\mu-\bar s\geq 3\epsilon/10)\leq P\left(\mid \bar s-\mu\mid\geq \sigma\frac{3\sqrt{\Delta_u}\epsilon}{10}\right) \leq 2e^{-\frac{9\Delta_u\epsilon^2}{200}}\leq 2e^{-\frac{9\Delta\epsilon^2}{200}}.
\]
By a union bound argument, the probability that all rebels with degree at least $\Delta$ output ``many'' is at least $1-2ne^{-9\Delta\epsilon^2/200}$.  Since $\rho\geq 0.8$, we have that, for $n$ sufficiently large, with high probability (namely, at least $1-1/n^2$), the number of rebels with degree at least $\Delta$  is at least $n/3$. By using again a union bound argument, we have that the success probability is at least $1-2ne^{-9\Delta\epsilon^2/200}-1/n^2$.
	
For $\Delta$ sufficiently large, we have that $e^{-9\Delta\epsilon^2/200}\leq1/n^3$. Hence, the success probability is at least $1-O(1/n^2)$, and the total risk is at most $0.715\epsilon$, for $n$ sufficiently large. This concludes the proof of Theorem \ref{thm:success}.

\section{Analysis of the Median protocol}\label{sec:median}
The goal of this section is to prove Theorem~\ref{thm:median}.

As in the proof of Lemma~\ref{thm:safety}, let $P_0$ (respectively, $P_1$) denote the normal distribution with mean 0 (respectively, $\epsilon$) and variance 1. Recall that each obedient citizen sends to any other agent a signal drawn from $P_0$, while a rebel sends to any other agent a signal drawn from $P_1$. In the following, we say that a signal is {\em high} if it is above $\epsilon$, and, for any two agents $u$ and $v$, we denote by $\chi_{v\rightarrow u}$ the binary random variable such that $\chi_{v\rightarrow u}=1$ if and only if the signal sent by agent $v$ and received from agent $u$ is high. Clearly, we have that, if $v$ is a rebel, then $P(\chi_{v\rightarrow u}=1)=1/2$, while, if $v$ is an obedient citizen, then $P(\chi_{v\rightarrow u}=1)=\psi(\epsilon)$, where $\psi(z)=P(x\geq z)$ is the tail distribution of $P_0$. Let $\overline{\chi_u}$ be the random variable denoting the number of high signals received by an agent $u$, that is, $\overline{\chi_u}=\sum_{(v,u)\in E(G)} \chi_{v\rightarrow u}$. The expected value $ E[\overline{\chi_u}]$ of $\overline{\chi_u}$ satisfies the following inequalities:
\begin{equation}\label{eq:exp-lb}
E[\overline{\chi_u}]\geq\left((1/2)\rho + \psi(\epsilon)(1-\rho)-o(1)\right)\Delta_u,
\end{equation}
and
\begin{equation}\label{eq:exp-ub}
E[\overline{\chi_u}]\leq\left((1/2)\rho + \psi(\epsilon)(1-\rho)+o(1)\right)\Delta_u,
\end{equation}
where the term $o(1)$ accounts for the uncertainly that results from the behavior of undercover agents. A good approximation of the cumulative function of the normal distribution~\cite{Polya45,Aludaat08} implies the following upper and lower bounds for the tail distribution:
\begin{equation}\label{eq:lower-bound}
\psi(\epsilon)\geq\frac{1-\sqrt{1-\frac{1}{e^{\sqrt{\frac{\pi}{8}} \epsilon^2}}}}{2}-\frac{1}{500}>\frac{1}{2}\left(1-2\epsilon\right),
\end{equation}
and
\begin{equation}
\psi(\epsilon)\leq\frac{1-\sqrt{1-\frac{1}{e^{\sqrt{\frac{\pi}{8}} \epsilon^2}}}}{2}+\frac{1}{500} < \frac{1}{2}\left(1-\frac{2}{3}\epsilon\right),
\label{eq:upper-bound}
\end{equation}
where the second inequalities holds when $\epsilon\in[0.04,0.2]$. Equations~\ref{eq:exp-lb},~\ref{eq:exp-ub},~\ref{eq:lower-bound}, and~\ref{eq:upper-bound} imply that

\[
\label{eq:exp_many}
E\left[\overline{\chi_u}\mid\rho\geq 0.8\right] \geq \left(\frac{2}{5}+\frac{1}{5}\psi(\epsilon)-o(1)\right)\Delta_u\nonumber > \left(\frac{1}{2}-\frac{\epsilon}{5}-o(1)\right)\Delta_u,
\]
and that
\[
\label{eq:exp_few}
E\left[\overline{\chi_u}\mid\rho\leq 0.2\right] \leq \left(\frac{1}{10}+\frac{4}{5}\psi(\epsilon)+o(1)\right)\Delta_u\nonumber < \left(\frac{1}{2}-\frac{4}{15}\epsilon+o(1)\right)\Delta_u.
\]
We define $\varphi(\epsilon)=7\epsilon/30$ and $f(\epsilon)=1/2-\varphi(\epsilon)$ (note that $4\epsilon/15>\varphi(\epsilon)>\epsilon/5$). Let us now calculate the output-risk of a rebel $u$. This corresponds to the probability of having a false-positive mistake, that is, that $u$ outputs ``many'' even though $\rho\leq 0.2$. Let us therefore condition on having $\rho\leq 0.2$. If $\Delta_u<\Delta$ then the rebel $u$ does not output anything and hence there is no output-risk. Hence, aiming to given an upper bound on the output-risk, we may assume without loss of generality that $\Delta_u\geq\Delta$. In this case, agent $u$ outputs ``many'' if  $\overline{\chi_u}>f(\epsilon)\Delta_u$. We show next that for $\Delta$ sufficiently large, this happens with probability at most $e^{-0.015(1-2\epsilon)\epsilon^2\Delta}$. Indeed, for any $\rho\leq0.2$, the following holds: $P\left(\overline{\chi_u}>f(\epsilon)\Delta_u\right)<e^{-0.015(1-2\epsilon)\epsilon^2\Delta}$.

Hence, the output-risk is at most $e^{-0.015(1-2\epsilon)\epsilon^2\Delta}$. By Lemma  \ref{thm:safety}, the relative message-risk is $r_\mathrm{message}\leq \epsilon/\sqrt{2}$, hence, the total risk is
\[
r_\mathrm{total}=r_\mathrm{output}+r_\mathrm{message}\leq\frac{\epsilon}{\sqrt{2}}+e^{-0.015(1-2\epsilon)\epsilon^2\Delta}.
\]

Next, let us calculate the success probability. For this purpose, we need to lower bound the probability that at least a third of the rebels output ``many'' when $\rho\geq 0.8$. Consider a rebel $u$ with degree at least $\Delta$, i.e., $\Delta_u\geq \Delta$. For $\Delta$ sufficiently large, the probability that this rebel does not output ``many'' given that $\rho\geq 0.8$ is at most $e^{-0.036\epsilon^2\Delta}$. Indeed, for any $\rho\geq0.8$, the following holds: $P\left(\overline{\chi_u}<f(\epsilon)\Delta_u\right)<e^{-0.036\epsilon^2\Delta}$.

By a union bound argument, the probability that all rebels with degree at least $\Delta$ output ``many'' is at least $1-ne^{-0.036\epsilon^2\Delta}$. Since $\rho\geq 0.8$, we have that, with high probability, the number of rebels with degree at least $\Delta$  is at least $n/3$. By using again a union bound argument, we have that the success probability is at least $1-ne^{-0.036\epsilon^2\Delta}+1/n^2$.
	
Note that for $n$ and $\Delta$ sufficiently large, we get that the success probability is at least $1-O(1/n^2)$, and the total risk is at most $0.715\epsilon$, as stated in the theorem. This completes the proof of Proof of Theorem~\ref{thm:median}.

\section{An impossibility result under private communication}\label{sec:impossibility}
The goal of this section is to prove Theorem \ref{thm:impossibility}.

Consider the private communication framework in a regular network of degree $\Delta$ and size $n$, and let $\mathcal{R}$ be a deterministic uniform rebel protocol. Recall that in such a protocol, each rebel sends the same message $m$ to all its neighbors. Assume that the success probability of $\mathcal{R}$ is at least $p$, for some constant $p>0$ and for sufficiently large $n$ and $\Delta$.

We now design a specific police protocol $\mathcal{P}$, called the {\em Reverse police protocol} that confronts rebels with their own detection protocol. That is, in order to decide whether or not to arrest agent $u$, the reverse police protocol $\mathcal{P}$ does the following. First, it collects the set $\Lambda_u^{\rightarrow}$ of $\Delta$ signals corresponding to the $\Delta$ messages outgoing from $u$. All these signals are samples drawn from the same normal distribution, that is, either $N(0,1)$ or $N(m,1)$ depending on whether $u$ is an obedient citizen or a rebel. Note also that the signal in $\Lambda_u^{\rightarrow}$ viewed by the police, corresponding to the message sent by $u$ to one of its neighbor $v$ is not equal to the corresponding signal received by $v$; however, the two signals are both samples drawn from the same normal distribution (once again, either $N(0,1)$ or $N(m,1)$). The police protocol $\mathcal{P}$ then simulates, for each agent $u$, the rebel output decision protocol assuming that the incoming signals of $u$ are the ones in $\Lambda_u^{\rightarrow}$. In other words, $\mathcal{P}$ simulates the scenario in which agent $u$ is a rebel that receives from each of its neighbors $v$ the signal in $\Lambda_u^{\rightarrow}$ corresponding to the message sent by $u$ to $v$. If the simulation outputs ``many'' then the reverse police arrests agent $u$.  

In order to analyze the total risk $r_\mathrm{total}=r_\mathrm{output}+r_\mathrm{message}$ of a rebel, we will show that the message-risk of a rebel is at least $p/4$ by analyzing the case in which all agents are rebels (that is, $\rho=1$), and that the message-risk of an obedient citizen is at most $r_\mathrm{output}+1/n$ by analyzing the case in which almost all agents are obedient (for example, $\rho=1/(2n^2)$). Since the relative message-risk $r_\mathrm{message}$ of a rebel is the difference between the message-risk of a rebel and the message-risk of an obedient citizen, we have that $r_\mathrm{message}\geq p/4-r_\mathrm{output}-1/n$. Hence,  $r_\mathrm{total}\geq r_\mathrm{output}+p/4-r_\mathrm{output}-1/n=p/4-1/n$, and the theorem is thus proved. It then remains to prove the two bounds on the message-risk of a rebel and of an obedient citizen.

\subsection*{Bounding the message-risk of a rebel.}

Let us consider the case in which all agents are rebels, that is, $\rho=1$. In this case, by the correctness guarantee, with probability at least $p$, at least $1/3$ of the agents output ``many''. Let $\hat{p}_{n,\Delta}$ be the probability that a rebel outputs ``many'' (assuming a $\Delta$-regular network of size $n$ where all agents are rebels). Hence, the expected fraction of rebels that output ``many'' is $\hat{p}_{n,\Delta}$. Note that, since $\rho=1$, each rebel outputs ``many'' with probability $\hat{p}_{n,\Delta}$ independently of other agents. This is because each agent (being a rebel) sends the same message $m$ to each of its neighbors. The signals in the system are therefore independent samples taken from $N(m,1)$. The output decision of each rebel is based on its incoming signals which are independent from the incoming signals of other rebels. Because of the Markov inequality, we have that there exists some $n_0$ such that, for every $n>n_0$, $\hat{p}_{n,\Delta}\geq p/4$. Indeed, let $X_u$ be the random binary variable indicating whether a rebel $u$ outputs ``many'', and let $X=\sum X_u$. Then $\mu=E[X]=\sum E[X_u]=n\hat{p}_{n,\Delta}$. By contradiction, suppose that $\hat{p}_{n,\Delta}<p/4$. By applying Markov's inequality, we have that
\[P(X\geq n/3)\leq n\hat{p}_{n,\Delta}/(n/3)<3p/4<p,\]
contradicting the hypothesis that there exists some $n_0$ such that, for every $n>n_0$, the success probability is at least $p$. In what follows we hence assume $n>n_0$. 

The decision of rebel $u$ is based on the collection of received signals $\Lambda_u^{\leftarrow}$. As all agents are rebels, as mentioned, all signals follow the same distribution $N(m,1)$. Hence, the signals in $\Lambda_u^{\leftarrow}$ are samples drawn from the same distribution as the signals in $\Lambda_u^{\rightarrow}$ (that is, the set of signals collected by the police protocol). Since a rebel whose input is $\Lambda_u^{\leftarrow}$ outputs ``many'' with probability at least $p/4$, it also outputs ``many'' with probability at least $p/4$ when given as input the collection $\Lambda_u^{\rightarrow}$. By the definition of the Reverse police protocol $\mathcal{P}$, it follows that the police arrests a rebel with probability at least $p/4$. In other words, the message-risk of a rebel is  at least $p/4$. 

\subsection*{Bounding the message-risk of an obedient citizen.} 
Let us consider the case in which almost all agents are obedient citizens, that is, the case in which $\rho$ tends to zero. More precisely, assume that $\rho=1/(2n^2)$. In particular, for  $n\geq 2$, we have $\rho<0.2$. The event that there is at least one rebel in the network happens with positive probability. Consider such a rebel $r$. By the definition of output-risk, based on the set of its incoming signals $\Lambda_r^{\leftarrow}$, rebel $r$   outputs ``many'' with probability at most $r_\mathrm{output}$.  Let $A$ be the event that all the $\Delta$ neighbors of $r$ are obedient citizens. Then 
\[
r_\mathrm{output}\geq P(r \mbox{~outputs ~``many''})\geq P(r \mbox{~outputs ~``many''}\mid A) \cdot P(A). \]
By a union bound argument, the probability that at least one neighbor of $r$ is a rebel is at most $\Delta\rho$. Therefore, $P(A)\geq 1-\Delta\rho$. Together, we obtain:
\begin{equation}\label{eq_rebel_wrong}
\frac{r_\mathrm{output}}{1-\Delta\rho}\geq P(r \mbox{~outputs ~``many''}\mid A). 
\end{equation}
Next observe, that given $A$, the set of $\Delta$ incoming signals $\Lambda_r^{\leftarrow}$ follows the same distribution as  $\Lambda_u^{\rightarrow}$, namely, the set of $\Delta$ signals outgoing from an obedient citizen $u$. Hence, by the definition of $\mathcal{P}$, the police arrests an obedient citizen with probability $P(r \mbox{~outputs ~``many''}\mid A)$, which is, by Eq.\ref{eq_rebel_wrong}, at most \[\frac{r_\mathrm{output}}{1-\Delta\rho}<
r_\mathrm{output}\left(1+\frac{1}{n}\right)\leq r_\mathrm{output}+\frac{1}{n},
\]
for sufficiently large $n$. This completes the proof of Theorem \ref{thm:impossibility}.

\section{Analyzing the Self-Immolation protocol}\label{sec:analyze-self}

The goal of this section is to prove Theorem \ref{thm:self-mmulation}.

A rebel in the self-immolation protocol sends the message 0 with probability $1-q$, and hence, in this case it is indistinguishable from an obedient citizen. Thus, the expected relative message-risk of a rebel is at most $q$. 

Let us next estimate the output-risk of a rebel $u$. If the degree of $u$ is less than the median degree, i.e., $\Delta_u<\Delta$, then $u$ does not output anything and hence incurs an output-risk of zero. Let us therefore consider the case that  $\Delta_u\geq \Delta$. The probability that a neighbor $v$ of $u$ sends the message $M$ is $q\rho$ (that is, the probability that $v$ is a rebel times the probability that $v$ sends the message $M$). Let $X_u$ be the random variable denoting the number of neighbors of $u$ that send the message $M$. Hence, the expected value of $X_u$ is $\mu_u=q\rho\Delta_u$. In particular, if $\rho= 0.2$, then $\mu_u=(1/5)q\Delta_u$. Hence, using Chernoff's bound, the probability that a rebel outputs ``many'' given that $\rho\leq 0.2$ is
\[
P\left(X_u>\frac{c\Delta_u\log n}{2\Delta}\mid \rho\leq 0.2\right) \leq P\left(X_u>(1+3/2)\mu_u\mid \rho= 0.2\right) \leq e^{-\frac{3}{20}q\Delta} = e^{-\frac{3c\log  n}{20}}= O(1/n^3),
\]
for a sufficiently large constant $c$. This implies that the output-risk is $O(1/n^3)$, and the expected total risk is $O(q+1/n^3)=O(\log n/\Delta+1/n^3)=O(\log n/\Delta)$, as stated in the theorem.

Next, let us calculate the success probability. For this purpose, we need to lower bound the probability that at least a third of the rebels output ``many'' when $\rho\geq 0.8$. Consider a rebel $u$ with degree at least $\Delta$. Using similar arguments to the ones mentioned above, the probability that this rebel does not output ``many'' given that $\rho\geq 0.8$ is 
\[
P\left(X_u<\frac{c\Delta_u\log n}{2\Delta}\mid \rho\geq 0.8\right)=O(1/n^3).
\]
By a union bound argument, the probability that all rebels with degree at least $\Delta$ output ``many'' is at least $1-O(1/n^2)$. Since $\rho\geq 0.8$, we have that, with high probability (namely, at least  $1-O(1/n^2)$), the number of rebels with degree at least $\Delta$ is at least $n/3$. By using a union bound argument, we have that the success probability is at least $1-O(1/n^2)$.
This completes the proof of Theorem \ref{thm:self-mmulation}.

\section{Discussion}\label{sec:discussio}

This paper argues that the communication infrastructure can play a significant role in rebels' ability to estimate their fraction in the population securely. Our main takeaway message is that even under extreme surveillance conditions, there are simple deterministic protocols in the public communication model that allow rebels to estimate their fraction in the population while keeping a negligible risk of each rebel being identified as such. In light of these results, it may be interesting to revisit the emergence of past revolutions, especially in non-democratic countries. 
For example, a pivotal moment in the Romanian revolution was the botched public speech that Ceau\cb{s}escu gave on 21 December 1989. In the wake of growing social tension, Ceau\cb{s}escu conducted a speech before a crowd consisting of tens of thousands in Palace Square. Aiming to demonstrate the control of the leader, the speech was nationally televised to millions. The crowd were given orders on  when to applaud and what to chant, while secrete policemen were among the crowd making sure that everything is in order. Such a speech was conducted yearly, but this time something different happened. In the beginning of the speech, the crowd stayed quite when Ceau\cb{s}escu speaks, applauded the leader at intermediate pauses and chanted admiration songs. However, eight minutes into the speech, some sound began to arise from the crowd which became louder and louder until the crowd starting booing. Ceau\cb{s}escu and his wife Elena fled the scene by helicopter; A day after they  were captured, put on trial, and shot by a firing squad. It is unclear what dynamics led to the dramatic switch in the crowd's behavior, from completely submissive to rebellious. Among other factors, it appears plausible that despite the surveillance, rebels in the crowd managed to somehow understand that if they suddenly act rebelliously, sufficiently many others would join. This paper suggests that the public pattern of communication that is inherent to crowd assemblies, could have had a non-negligible contribution to the emergent changeover in the crowd's behavior.

The principle revealed in this paper is in fact not limited to overthrows of dictatorships by rebels and can be pertinent to other social movements in which the participating individuals prefer to remain covert. Interestingly, yet more speculatively, our setting may further find relevance in the microbiological world. Indeed, quorum-sensing mechanisms are known to be utilized by bacteria communities to identify when their density passes a certain threshold \cite{quorum1,quorum2,quorum3,quorum4,quorum5}, e.g., before attacking a host tissue \cite{quorum2}. Moreover, communication between bacteria follows a diffusion process of autoinducers, which is, in some sense, reminiscent of public communication. In the presence of the immune system, it is plausible that pathogenic bacteria act covertly, especially while being surrounded by non-pathogenic bacteria communities. In this context, our results may suggest that in order to perform the quorum-sensing covertly, such bacteria would avoid using distinct autoinducers in their signaling, and instead, use signals composed of a mixture of molecules types that are already used by nearby nonpathogenic bacteria, while slightly biasing their proportions. A supporting empirical evidence is the fact that several common species of bacteria, including {\em B. subtilis}, {\em V. harveyi}, and its pathogenic relative, {\em V. cholerae}, have been shown to utilize different combinatorial combinations of autoinducers which are used (either separately or in other combinations) by other bacteria \cite{quorum3,quorum4,Auchtung,Henke,even2016transient,even2016social,bridges2019intragenus}. Explanations for the use of multiple autoinducers have been given using arguments from evolutionary game theory \cite{even2016transient,eldar2011social}. The current paper suggests that this phenomena, and particularly the overlap in the autoinducer combinations, could also be explained in the context of covert communication.

\begin{acks}
This is a post-peer-review, pre-copyedit version of an article published in \textit{Scientific Reports}. The final authenticated version is available online at: https://doi.org/10.1038/s41598-022-07165-9. We are particularly thankful for Simon Collet for proposing the idea behind the impossibility result (Theorem \ref{thm:impossibility}) and for Lucas Boczkowski for suggesting the self-immolation protocol (Theorem \ref{thm:self-mmulation}). In addition, we thank Ami Paz, Emanuele Natale, Jonathan Korman, Christos Papadimitriou, and Ofer Feinerman for helpful discussions. Finally, we thank Avigdor Eldar for helpful discussions regarding quorum-sensing mechanisms in bacteria. This work has received funding from the European Research Council (ERC) under the European Union's Horizon 2020 research and innovation program (grant  agreement No 648032).
\end{acks}

\bibliographystyle{ACM-Reference-Format}
\bibliography{korman_crescenzi.bib}

\end{document}